\documentclass[pra,floatfix,amsmath,twocolumn,showpacs]{revtex4}
\usepackage{amssymb}
\usepackage{graphicx}
\usepackage{graphics}
\usepackage{amsmath}
\usepackage{amsthm}
\usepackage{color}

\newcommand{\bea}{\begin{eqnarray}}
\newcommand{\eea}{\end{eqnarray}}

\def\bi{\begin{itemize}}
\def\ei{\end{itemize}}
\def\bc{\begin{center}}
\def\ec{\end{center}}

\def\C{\hbox{$\mit I$\kern-.7em$\mit C$}}
\def\R{\hbox{$\mit I$\kern-.6em$\mit R$}}

\newcommand{\one}{\mbox{$1 \hspace{-1.0mm}  {\bf l}$}}
\def\tr{\mathrm{tr}}

\def\ns{\mathcal{NS}}
\def\l{\mathcal{L}}
\def\q{\mathcal{Q}}
\def\P{\textbf{P}}

\newtheorem{theorem}{Theorem}

\newtheorem{proposition}[theorem]{Proposition}

\begin{document}

\title{Simple conditions constraining the set of quantum correlations}
\author{Julio I. de Vicente} \email{jdvicent@math.uc3m.es}
\affiliation{Departamento de Matem\'aticas, Universidad Carlos III de
Madrid, Avda. de la Universidad 30, E-28911, Legan\'es (Madrid), Spain}

\begin{abstract}

The characterization of the set of quantum correlations in Bell scenarios is a problem of paramount importance for both the foundations of quantum mechanics and quantum information processing in the device-independent scenario. However, a clear-cut (physical or mathematical) characterization of this set remains elusive and many of its properties are still unknown. We provide here simple and general analytical conditions that are necessary for an arbitrary bipartite behaviour to be quantum. Although the conditions are not sufficient, we illustrate the strength and non-triviality of these conditions with a few examples. Moreover, we provide several applications of this result: we prove a quantitative separation of the quantum set from extremal nonlocal no-signaling behaviours in several general scenarios, we provide a relation to obtain Tsirelson bounds for arbitrary Bell inequalities and a construction of Bell expressions whose maximal quantum value is attained by a maximally entangled state of any given dimension.
\end{abstract}


\maketitle

\section{Introduction}

Bell's theorem shows that the correlations exhibited by quantum mechanical systems go beyond what is achievable by any local realistic theory (i.\ e.\ any local hidden variable model) \cite{bell}. This result has deep implications and still spurs several subfields of research \cite{review}. On the one hand, from a foundational perspective, it has been questioned what physical principle in nature could then be responsible for giving rise to precisely the particular set of correlations allowed by quantum mechanics. Since it has been shown that this set is strictly smaller than what could be achieved by just imposing the no-signaling principle (i.\ e.\ the impossibility of instantaneous propagation of information) \cite{tsins,PR}, several works have considered more restrictive physically-motivated axioms \cite{ic,ml,lo}. Although these approaches rule out several subsets of no-signaling correlations, no definite answer has yet been found as there still exist supra-quantum correlations compatible with these principles \cite{almostquantum}. On the other hand, from a more practical point of view and in the context of quantum information theory, it has been realized that quantum nonlocality can be regarded as a resource for device-independent quantum information processing (DIQIP). The tasks that can be carried out in this way include key distribution for cryptography \cite{qkd}, randomness generation \cite{re} and dimensionality certification (see e.\ g.\ \cite{dim,diment}).

In order to understand which physical principles constrain the set of quantum correlations and to elucidate what are the ultimate limitations behind DIQIP protocols, a fundamental question arises: can an efficient mathematical description of this set be found? It turns out that this is a highly nontrivial question. So far, the only systematic and general form to bound this set is given by the NPA hierarchy \cite{NPA}. This provides a hierarchy of semi-definite programs that approximate the quantum set from the outside. 
Although this constitutes an extremely powerful tool, it has been mainly exploited numerically. Thus, many properties of the quantum set remain unknown, since verifying general statements cannot be usually approached by numerical means. It would be therefore desirable to have simple analytical conditions constraining the set of quantum correlations. These could be moreover used to exclude some general subsets of no-signaling correlations from the quantum set based on their analytical properties or to bound the maximal efficiency DIQIP tasks with a given structure can attain
. The aim of this article is to fill this gap: I provide simple general analytical conditions any quantum behaviour should satisfy which rely on standard matrix analysis. Although these conditions emerge from the first step of the NPA hierarchy and are, therefore, only necessary, I will consider examples showing their strength and non-triviality. Furthermore, I will provide several different applications of this result: a proof of the separation between the quantum set and extremal nonlocal no-signaling correlations (a question recently raised in \cite{ravi}), a systematic way to obtain quantum bounds on arbitrary Bell inequalities (generalizing the recent result \cite{epping}) and the possibility to do robust self-testing of bipartite maximally entangled states by building Bell inequalities that are maximally violated by them.

\section{Framework and main result}

We will consider the standard bipartite $(m_Am_Bd_Ad_B)$ Bell scenario \cite{review} in which two parties $A$ and $B$ that could have interacted in the past remain now uncommunicated. Party $A$ ($B$) can freely choose questions from a finite alphabet $\mathcal{X}=\{1,2,\ldots,m_A\}$ ($\mathcal{Y}=\{1,2,\ldots,m_B\}$). Given any of these inputs $x\in\mathcal{X}$ and $y\in\mathcal{Y}$, each party can obtain an outcome $a\in\mathcal{A}=\{1,2,\ldots,d_A\}$ and $b\in\mathcal{B}=\{1,2,\ldots,d_B\}$, where the output alphabets are also finite. The central object here, referred to as behaviour, is the joint conditional probability distribution of obtaining the outputs $(a,b)$ given the choice of inputs $(x,y)$, $P(ab|xy)$ (for which we will use the shorthand $\P$). This list of $d_Ad_Bm_Am_B$ numbers must fulfill $P(ab|xy)\geq0$ $\forall a,b,x,y$ and $\sum_{a,b}P(ab|xy)=1$ $\forall x,y$. Moreover, since communication among the parties is not possible during the choice of input and recording of the output, the marginal of each party must be independent of the other's action, $P(a|x)=\sum_bP(ab|xy)=\sum_bP(ab|xy')$ $\forall a,x,y\neq y'$ and $P(b|y)=\sum_aP(ab|xy)=\sum_aP(ab|x'y)$ $\forall b,y,x\neq x'$. The set of all behaviours satisfying these no-signalling conditions will be denoted by $\mathcal{NS}$. Particular elements of this set are the $d_A^{m_A}d_B^{m_B}$ different local deterministic behaviours (LDBs) $D_i(ab|xy)=\delta_{a,f_i(x)}\delta_{b,g_i(y)}$ in which for every party a unique output occurs with probability 1 for every choice of input. The convex hull of these behaviours gives rise to the local set $\mathcal{L}$ \cite{fine}. On the other hand, the quantum set $\mathcal{Q}$ is given by all behaviours that can be obtained by performing measurements on bipartite quantum states of unrestricted dimension $\rho_{AB}$, i.\ e.\ $P(ab|xy)=\tr(\rho_{AB}E_a^x\otimes F_b^y)$ for some projectors $\{E_a^x,F_b^y\}$ such that $\sum_aE_a^x$ and $\sum_bF_b^y$ equal the identity in each party's Hilbert space $\forall x,y$. The crucial observation mentioned in the introduction is that $\l\subsetneq\q\subsetneq\ns$. Although $\q$ is a convex set, it is in general very hard to decide from the definition whether a given behaviour is in $\q$ or not. On the contrary, $\l$ and $\ns$ are both convex polytopes with vertices given by the LDBs in the first case, to which we have to add some non-local vertices in the second case. Following the standard notation, we will refer to these extremal non-local behaviours as Popescu-Rohrlich (PR) boxes.

In order to present our results, we will use some further notation.
We will arrange every $\P\in\ns$ to form the $m_Ad_A\times m_Bd_B=n_A\times n_B$ real matrix
\begin{equation}\label{Pmatrix}
P=\sum_{abxy}P(ab|xy)|xa\rangle\langle yb|,
\end{equation}
where in the standard notation of quantum mechanics $|xa\rangle=|x\rangle\otimes|a\rangle$ and $\{|x\rangle\}$ denotes the computational basis of $\mathbb{R}^{m_A}$ and similarly for the other alphabet elements. Thus, $P$ can be partitioned as a block matrix
$$P=\left(
      \begin{array}{ccc}
        P_{11} & \cdots & P_{1m_B} \\
        \vdots & \ddots & \vdots \\
        P_{m_A1} & \cdots & P_{m_Am_B} \\
      \end{array}
    \right)\in\mathbb{R}^{m_Ad_A\times m_Bd_B}=\mathbb{R}^{n_A\times n_B}$$
with blocks
$$P_{xy}=\left(
                           \begin{array}{ccc}
                             P(11|xy) & \cdots & P(1d_B|xy) \\
                             \vdots & \ddots & \vdots \\
                             P(d_A1|xy) & \cdots & P(d_Ad_B|xy) \\
                           \end{array}
                         \right)\in\mathbb{R}^{d_A\times d_B}.$$
Normalization imposes that the entries in each block add up to one while no-signaling that the sum of the elements in the same row (column) for blocks in the same row (column) is equal. We will furthermore consider different matrix norms. Following the Schatten p-norm notation, $||\cdot||_1$ will be the trace norm (i.\ e.\ the sum of all singular values) while $||\cdot||_\infty$ the spectral norm (i.\ e.\ the maximal singular value). We are now in the position to state our main result.

\begin{theorem} In every $(m_Am_Bd_Ad_B)$ scenario, if $\P\in\q$ then $||P||_1\leq\sqrt{m_Am_B}$.
\end{theorem}
We will actually prove the following stronger result: for every $\P\in\q$ and $G\in\mathbb{R}^{n_A\times n_B}$ it must hold that
\begin{align}
\langle P,G\rangle&=\tr(PG^T)=\sum_{abxy}P(ab|xy)G(ab|xy)\nonumber\\
&\leq||G||_\infty\sqrt{m_Am_B}.\label{ineq1}
\end{align}
The relevance of this result will be discussed later on. For the moment, let us simply point out that Theorem 1 follows from it by noticing that $||P||_1=\max_O\tr(PO)$ where the maximization is over all isometries \cite{HJ2}. 
\begin{proof}
In order to prove that inequality (\ref{ineq1}) is true, we will show that the inequality holds $\forall\P\in Q^1$, i.\ e.\ the first step of the NPA hierarchy (notice that $\q\subset Q^1$). It is worth mentioning that Theorem 1 can be proved without invoking this \cite{referee}. However, we have chosen to present this proof because it furthermore clarifies the relative strength of the condition. Behaviours in $Q^1$ must fulfill \cite{NPA} that a $(n_A+n_B)\times(n_A+n_B)$ real positive semidefinite matrix $\Gamma=\left(\begin{array}{cc}
Q & P \\
P^T & R \\
\end{array}\right)$ exists with $Q_{ii}=P_A(i)$ and $R_{jj}=P_B(j)$, where $P_A$ ($P_B$) is a vector of $\mathbb{R}^{m_Ad_A}$ ($\mathbb{R}^{m_Bd_B}$) with entries given by $P(a|x)$ ($P(b|y)$) with $xa$ ($by$) in lexicographical order \footnote{Other entries of $Q$ and $R$ have to be fixed to zero. See \cite{NPA} for details.}. Thus, defining $W=\left(
                               \begin{array}{cc}
                                 0 & G \\
                                 G^T & 0 \\
                               \end{array}
                             \right)$,
for any given $G$ the maximum value of $\langle P,G\rangle$ attainable in $\q$ cannot be larger than
\begin{align}
&\max_{\Gamma}\tr(\Gamma W)/2, \textrm{ subject to}\nonumber\\
&\Gamma\geq0,\nonumber\\
&\tr(D^A_i\Gamma)=P_A(i) \quad(i=1,\ldots,n_A),\nonumber\\
&\tr(D^B_j\Gamma)=P_B(j)\quad(j=1,\ldots,n_B),\nonumber
\end{align}
where $D^A_i=E_i\oplus 0_{n_B}$ ($D^B_j=0_{n_A}\oplus\tilde{E}_j$) and $E_i$ ($\tilde{E}_j$) a $n_A\times n_A$ ($n_B\times n_B$) matrix whose only nonzero entry is the $ii$ ($jj$) with value 1. This is the primal form of a semi-definite program (SDP) \cite{SDP} with cost function $p(\Gamma)$ and we will denote its solution by $p(\Gamma^*)$. The dual form of this SDP corresponds to
$$\min_x x^T\left(
              \begin{array}{c}
                P_A \\
                P_B \\
              \end{array}
            \right), \textrm{ subject to } \phantom{x}diag(x)-W/2\geq0,$$
where $x$ is a $n_A+n_B$ real vector yielding the value $d(x)$. By duality, for any feasible $x$ (i.\ e.\ satisfying the constraint above), it must hold that $p(\Gamma^*)\leq d(x)$. Thus, to finish the proof it suffices to construct a feasible $x$ yielding the value $d(x)=||G||_\infty\sqrt{m_Am_B}$. This is the case for $x=||G||_\infty/2(\sqrt{m_B/m_A}\vec{1}_{n_A},\sqrt{m_A/m_B}\vec{1}_{n_B})$ where $\vec{1}_{n}$ is the $n$-dimensional vector with all entries equal to one. To see that it is feasible amounts to checking that
$$\left(
    \begin{array}{cc}
      \sqrt{m_B/m_A}||G||_\infty\one_{n_A} & -G \\
      -G^T & \sqrt{m_A/m_B}||G||_\infty\one_{n_B} \\
    \end{array}
  \right)\geq0.$$
This is indeed true because, since the upper left corner in strictly positive definite, by Schur's complement condition \cite{HJ1} this is equivalent to
$$\sqrt{m_A/m_B}\left(||G||_\infty\one_{n_B}-\frac{G^TG}{||G||_\infty}\right)\geq0,$$
which is obviously true given that the maximal eigenvalue of $G^TG$ is precisely $||G||_\infty^2$.
\end{proof}

One of the appealing properties of Theorem 1 and inequality (\ref{ineq1}) is that they have a very compact form. However,the reasoning used in this proof can be applied to obtain other stronger but more complicated conditions \footnote{As we will see later this line of thought can also be applied in the correlation picture.}.
Let us denote by $M$ the matrix constructed using the same prescription as $P$ in Eq.\ (\ref{Pmatrix}) but with entries given by
\begin{equation}\label{Mmatrix}
M(ab|xy)=P(ab|xy)-P(a|x)P(b|y).
\end{equation}
We will now prove that for any choice of matrix $G\in\mathbb{R}^{n_A\times n_B}$,
\begin{align}
\langle M,G\rangle\leq&||G||_\infty\sqrt{m_Am_B}-\frac{||G||_\infty}{2}\sqrt{\frac{m_B}{m_A}}\sum_{ax}
P(a|x)^2\nonumber\\
&-\frac{||G||_\infty}{2}\sqrt{\frac{m_A}{m_B}}\sum_{by}P(b|y)^2 \label{ineq2}
\end{align}
should hold $\forall \P\in\q$ and, hence
\begin{theorem} In every $(m_Am_Bd_Ad_B)$ scenario, if $\P\in\q$ then
$$||M||_1\leq\sqrt{m_Am_B}\left(1-\sum_{ax}
\frac{P(a|x)^2}{2m_A}-\sum_{by}\frac{P(b|y)^2}{2m_B}\right).$$
\end{theorem}
\begin{proof}
This proof goes along similar lines as the previous one, so we will only sketch the details.
The set $Q^1$ is actually equivalent to the positive semidefiniteness of $\Gamma-\left(
\begin{array}{c}
P_A \\
P_B \\
\end{array}                                                                                         \right)(P_A^T P_B^T)$ \cite{ml}, that is, $$\tilde{\Gamma}=\left(\begin{array}{cc}
\tilde{Q} & M \\
M^T & \tilde{R} \\
\end{array}\right)\geq0,$$where now $\tilde{Q}_{ii}=P_A(i)-P_A(i)^2$ and $\tilde{R}_{jj}=P_B(j)-P_B(j)^2$. Thus, the value of $\langle M,G\rangle$ in $\q$ cannot be larger than the primal SDP
\begin{align}
&\max_{\tilde{\Gamma}}\tr(\tilde{\Gamma} W)/2, \textrm{ subject to}\nonumber\\
&\tilde{\Gamma}\geq0,\nonumber\\
&\tr(D^A_i\tilde{\Gamma})=P_A(i)-P_A(i)^2 \quad(i=1,\ldots,n_A),\nonumber\\
&\tr(D^B_j\tilde{\Gamma})=P_B(j)-P_B(j)^2\quad(j=1,\ldots,n_B),\nonumber
\end{align}
with dual
$$\min_x x^T\left(
              \begin{array}{c}
                P_A-P_A^2 \\
                P_B-P_B^2 \\
              \end{array}
            \right), \textrm{ subject to } \phantom{x}diag(x)-W/2\geq0,$$
where the vectors $P_{A,B}^2$ have entries $P_A(i)^2$ and $P_B(j)^2$. The same choice of $x$ as in the previous proof does the job.
\end{proof}

Let us finish this section analyzing the strength of Theorems 1 and 2 by considering a few examples. As mentioned before, our conditions are deduced from the definition of the first step of the NPA hierarchy, $Q^1$. Obviously then, the best we can expect from them is to separate this set from its complement in $\ns$. Therefore, since $\q\subsetneq Q^1$, it is clear that there exist supra-quantum behaviours satisfying the conditions of Theorems 1 and 2 (i.\ e.\ they are necessary for a behaviour to be in $\q$ but not sufficient). 
We have performed numerical explorations in the (2222) scenario that show that the gap between $Q^1$ and the behaviours satisfying our conditions is small. Theorem 2 only provides a slightly better approximation of $Q^1$ than Theorem 1. A more detailed example can be found in Figure 1. To my knowledge, the only previous instance of analytical means to constrain $\q$ is given by the results of \cite{anal} (which emerge from $Q^1$ as well). However, the application of these techniques is not completely straightforward as they rely on some choice of data processing. Still, the inequality emerging from this approach reproduces $Q^1$ in the (2222) scenario. Notwithstanding, Figure 1 shows that our conditions give a much tighter restriction already in the (3322) example considered in \cite{anal}. This is also apparent in the (2233) case, where the results presented in \cite{anal} fail to completely reproduce $Q^1$ for the isotropic behaviours obtained by mixtures of a fully random box and the PR box $P_{PR}(2,3)$ (see Eq.\ (\ref{pr2d}) below). On the contrary, Theorems 1 and 2 are tight in this case. Thus, it is interesting to point out that while inequalities (\ref{ineq1}) and (\ref{ineq2}) depend on a proper choice of $G$ for each case, these examples indicate that Theorems 1 and 2, which are straightforward to apply in general, still provide remarkably strong conditions. To show further the usefulness of these results, in the remaining sections we provide several applications of them.

\begin{figure}[t]
\includegraphics[scale=0.4]{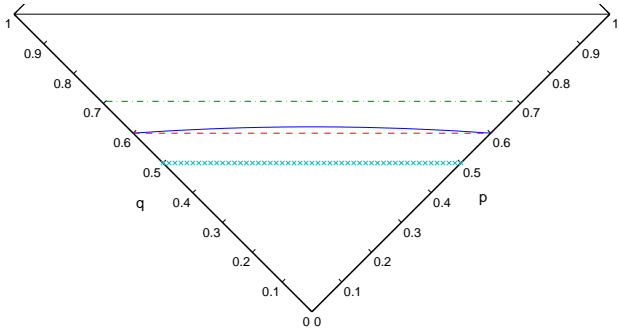}
\caption{A slice of the (3322) polytope spanned by behaviours $P=qP_1+pP_2+(1-p-q)P_n$, where $P_1$ and $P_2$ are two PR boxes and $P_n$ is the fully mixed behaviour. Horizontal lines achieve the same value for the $I_{3322}$ Bell inequality (see \cite{anal} for details). Crosses: The boundary of $\l$ and $\q$ (equal in this slice, $I_{3322}=0$). Dashed line: The boundary of $Q^1$ ($I_{3322}=0.2$). Solid line: The conditions of Theorems 1 and 2 (they are equivalent in this case, $0.2\leq I_{3322}\lesssim0.24$). Dashed-dotted line: The inequality of \cite{anal} with the sign binning data processing ($I_{3322}\simeq0.41$). A more clever data processing for this case only manages to lower this condition to $I_{3322}\simeq0.3$ \cite{anal}.}
\end{figure}

\section{Nontriviality of the conditions and nonquantumness of extremal no-signaling behaviours}\label{nontsec}

The previous examples already give a good idea of the strength of the conditions derived here. To analyze in full generality their non-triviality, notice that, due to the convexity of the trace norm, $||P||_1$ must attain its maximum value in $\ns$ at the vertices of the polytope. Hence, the ideal situation would be that the LDBs achieve the maximal possible value ($\sqrt{m_Am_B}$) and that all PR boxes violate this bound. Actually, it has been only recently shown in \cite{ravi} (see also \cite{lo}) that all PR boxes (including the multipartite case) are not in $\q$. It was left as an open question there whether there exists a separation between them and $\q$. The interest of this question relies on the fact that PR boxes are the most advantageous resources in DIQIP applications such as cryptography \cite{ravi}. Thus, the closer $\q$ can be to these behaviours the more efficient these applications can be. In the following we show that all LDBs attain the $\sqrt{m_Am_B}$ bound and that all PR boxes in $(22d_Ad_B)$ and $(mm22)$ scenarios violate the bound. Thus, besides showing that the bound is not improvable and in general not trivial, we further provide a simple proof in these cases of the result of \cite{ravi}. Moreover, we show that there actually exists a quantitative separation between these PR boxes and $\q$, answering in these scenarios the question raised therein.

\begin{proposition}
In every $(m_Am_Bd_Ad_B)$ scenario, if $\P\in\l$ then $||P||_1\leq\sqrt{m_Am_B}$, with equality for LDBs.
\end{proposition}
\begin{proof}
Every LDB is of the form $D_i(ab|xy)=d_i^A(a|x)d_i^B(b|y)$ where the lists $d_i^A(a|x)$ ($d_i^B(b|y)$) have $m_A$ ($m_B$) entries equal to 1 and 0 otherwise. Hence,
\begin{align}
||D_i||_1&=\left|\left|\left(\sum_{ax}d_i^A(a|x)|ax\rangle\right)\left(\sum_{by}d_i^B(b|y)\langle by|\right)\right|\right|_1\nonumber\\
&=\left|\left|\sum_{ax}d_i^A(a|x)|ax\rangle\right|\right|_2\left|\left|\sum_{by}d_i^B(b|y)|by\rangle\right|\right|_2\nonumber\\
&=\sqrt{m_Am_B}.
\end{align}
Notice that the fact that $||P||_1\leq\sqrt{m_Am_B}$ holds $\forall\P\in\l$ follows then by the convexity of the trace norm without the need of invoking Theorem 1.
\end{proof}
Evidently, that the inequality is fulfilled in $\l$ was already obvious from Theorem 1. The important observation here is that all LDBs attain the bound, hence showing that it cannot be improved.

We analyze now the values $||P||_1$ might take for PR boxes. Let us consider first the $(22d_Ad_B)$ scenario. Taking $A_d$ to be the circulant $d\times d$ matrix
\begin{equation}\label{Ad}
A_d=\left(
                                 \begin{array}{cccc}
                                   0 & 0 & \ldots & 1 \\
                                   1 & \ddots & \ddots & \vdots \\
                                   \begin{array}{c}
                                     0 \\
                                     \vdots
                                   \end{array}
                                    & \ddots & \ddots & 0 \\
                                   0 & \ldots0 & 1 & 0 \\
                                 \end{array}
                               \right),
\end{equation}
the PR boxes in this case are given by \footnote{We have removed extra rows or columns of zeros since they are irrelevant for the value of the trace norm.}
\begin{equation}\label{pr2d}
P_{PR}(2,d)=\frac{1}{d}\left(
              \begin{array}{cc}
                \one_d & \one_d \\
                \one_d & A_d \\
              \end{array}
            \right), \quad2\leq d\leq\min(d_A,d_B),
\end{equation}
up to relabelings of the inputs and the outputs \cite{barrett}. Since these transformations amount to certain permutations of the rows or columns of $P$ that leave the trace norm invariant, it suffices to compute it for the matrix given in Eq.\ (\ref{pr2d}). Using the pinching inequality \cite{bhatia,gohberg}, we obtain that
$$||P_{PR}(2,d)||_1\geq\frac{1}{d}\left\|\left(
              \begin{array}{cc}
                \one_d & 0_d \\
                0_d & A_d \\
              \end{array}
            \right)\right\|_1=2.$$
The conditions under which equality is attained in the pinching inequality are given in Theorem 8.7 of \cite{gohberg} and it is easily checked that they are not met in this case. Hence, we obtain that $||P_{PR}(2,d)||_1>2$, which amounts to the non-quantumness of these PR boxes by Theorem 1. Moreover, by a more refined use of the pinching inequality, we obtain the following stronger result, which shows the existence of a finite gap between these PR boxes and $\q$.

\begin{theorem}
In every $(22d_Ad_B)$ scenario, it holds $\forall \P\in\q$ that $\forall d$ $||P_{PR}(2,d)-P||_1\geq||P_{PR}(2,d)||_1-||P||_1\geq\sqrt{5}-2$.
\end{theorem}
\begin{proof}
We need to show that $||P_{PR}(2,d)||_1\geq\sqrt{5}$ $\forall d$. By permutation matrices, that leave the trace norm invariant, we can map $P=\sum_{abxy}P(ab|xy)|xa\rangle\langle yb|$ to $P'=\sum_{abxy}P(ab|xy)|ax\rangle\langle by|$, which has now blocks given by
$$P'_{ab}=\left(
                           \begin{array}{ccc}
                             P(ab|11) & \cdots & P(ab|1m_B) \\
                             \vdots & \ddots & \vdots \\
                             P(ab|m_A1) & \cdots & P(ab|m_Am_B) \\
                           \end{array}
                         \right)\in\mathbb{R}^{m_A\times m_B}.$$
We therefore have that
\begin{align}
||P_{PR}(2,d)||_1&=||P'_{PR}(2,d)||_1\nonumber\\
&=\frac{1}{d}\left|\left|\left(
                                 \begin{array}{cccc}
                                   Z & 0_2 & \ldots & Y \\
                                   Y & \ddots & \ddots & \vdots \\
                                   \begin{array}{c}
                                     0_2 \\
                                     \vdots
                                   \end{array}
                                    & \ddots & \ddots & 0_2 \\
                                   0_2 & \ldots0_2 & Y & Z \\
                                 \end{array}
                               \right)\right|\right|_1,
\end{align}
where
$$Z=\left(
      \begin{array}{cc}
        1 & 1 \\
        1 & 0 \\
      \end{array}
    \right)
,\quad Y=\left(
              \begin{array}{cc}
                0 & 0 \\
                0 & 1 \\
              \end{array}
            \right).$$
Hence, using now the pinching inequality we obtain the desired result $||P_{PR}(2,d)||_1\geq(d||Z||_1)/d=\sqrt{5}$.
\end{proof}
It might be interesting to note that $||P_{PR}(2,d)||_1\leq2\sqrt{2}$. To see this, notice that for the Frobenius norm $||P_{PR}(2,d)||_2=\sqrt{\tr(P_{PR}(2,d)P_{PR}(2,d)^T)}=2/\sqrt{d}$ and that $||X||_1\leq\sqrt{n}||X||_2$ for any $n\times n$ matrix $X$. Numerics suggest that the above estimates can be improved to $1+\sqrt{2}=||P_{PR}(2,2)||_1\leq||P_{PR}(2,d)||_1<\lim_{d\to\infty}||P_{PR}(2,d)||_1\simeq2.55$, that would change the bound in Theorem 4 to $\sqrt{2}-1$.

Let us move now to the $(mm22)$ scenario. The corresponding PR boxes have all been determined in \cite{JM} (see Table II therein). We denote an arbitrary one of them by $P_{PR}(m,2)$. One sees that (up to relabelings) these matrices always have the following structure: they have a $4\times4$ block in the diagonal given by $P_{PR}(2,2)$ followed by $m-2$ $2\times2$ blocks in the diagonal, which are either $\one_2/2$, $A_2/2$ or $diag(1,0)$. Since the latter blocks all have unit trace norm, it follows again by the pinching inequality that $||P_{PR}(m,2)||_1\geq||P_{PR}(2,2)||_1+m-2=m+\sqrt{2}-1$. This shows again by virtue of Theorem 1 that there is a quantitative separation between the PR boxes in these scenarios and $\q$:
\begin{theorem}
In every $(mm22)$ scenario, it holds $\forall \P\in\q$ that $||P_{PR}(m,2)-P||_1\geq\sqrt{2}-1$.
\end{theorem}

It might be interesting to mention as well that in the $(mm22)$ scenario for the fully nondeterministic boxes it holds that $||P_{PR}(m,2)||_1\leq m\sqrt{m}$. For these behaviours there are $2m^2$ non-vanishing entries with value $1/2$ \cite{JM}. Hence, similarly as before, $||P_{PR}(m,2)||_2=m/\sqrt{2}$, obtaining the above estimate.

\section{Tsirelson bounds and bipartite maximally entangled states}\label{tsisec}

One can see that the left-hand-side of inequality (\ref{ineq1}) defines an arbitrary Bell expression, i.\ e.\ any linear combination of the elements $P(ab|xy)$. Since $\l$, $\q$ and $\ns$ are compact convex sets, there always exist such expressions separating them, i.\ e.\ $\langle P,G\rangle\leq G_\l,G_\q,G_\ns$ depending on whether $\P\in\l,\q,\ns$ with $G_\l\leq G_\q\leq G_\ns$. The most characteristic one is the CHSH inequality (see below) in the $(2222)$ scenario for which $G_\l=2$ \cite{chsh}, $G_\q=2\sqrt{2}$ \cite{tsirelson} and $G_\ns=4$ \cite{PR}. While to determine the optimal value of $G_\l$ and $G_\ns$ it suffices to check over the corresponding vertices, to determine the optimal value of $G_\q$, known as Tsirelson bounds, is a less straightforward task \footnote{Still, determining $G_\l$ and $G_\ns$ can be computationally hard as the number of vertices increases exponentially with the number of inputs.}.
However, this is very relevant to identify optimal DIQIP performances in the context of quantum games \cite{review,review2}. Thus, a remarkable feature of inequality (\ref{ineq1}) is that it provides a systematic way to construct quantum upper bounds to arbitrary Bell inequalities. Actually, our result resembles that of \cite{epping,linden}, but the latter only holds for the particular class of Bell inequalities based on correlators. Later on, we will discuss this relation further. To give a hint of the usefulness of inequality (\ref{ineq1}), we will show now that it allows to obtain Tsirelson's bound for the CHSH inequality. This can be expressed by a matrix $G_{CHSH}$ with blocks
$$G_{11}=G_{12}=G_{21}=-G_{22}=\left(
\begin{array}{rr}
1 & -1 \\
-1 & 1 \\
\end{array}\right).$$
It turns out that $||G_{CHSH}||_\infty=2$ and, hence, we obtain the trivial bound 4. Nevertheless, given that behaviours in $\ns$ must fulfill several different constraints, equivalent Bell inequalities can be expressed up to rescaling and addition of an offset. Thus, if we take ($I$ is a matrix with all entries equal to 1)
$$G'_{CHSH}=\frac{1}{2}\left(
                         \begin{array}{cc}
                           I_2 & 0_2 \\
                           0_2 & I_2 \\
                         \end{array}
                       \right)
+\frac{1}{2\sqrt{2}}G_{CHSH},$$
it holds that $\langle P,G'_{CHSH}\rangle=1+\langle P,G_{CHSH}\rangle/(2\sqrt{2})$ $\forall \P\in\ns$. Since $||G'_{CHSH}||_\infty=1$, we have then that $\forall \P\in\q$, $\langle P,G'_{CHSH}\rangle\leq2$ and, hence, $\langle P,G_{CHSH}\rangle\leq2\sqrt{2}$.

Since the CHSH Tsirelson bound is achievable by a quantum behaviour arising from certain measurements on a maximally entangled two-qubit state and $G'_{CHSH}$ is orthogonal, this also shows that for this behaviour $||P||_1=2$, achieving the bound of Theorem 1. It is a natural question to ask which other behaviours in $\q\backslash\l$ can attain it. We computed $||P||_1$ for the quantum behaviours yielding the largest known value for several two-outcome Bell inequalities given in \cite{palvertesi} but, in general, the bound of Theorem 1 is not attained. Interestingly, when this occurs, the behaviour arises from a maximally entangled state of qubits. This seems to extend for scenarios with more outcomes. In particular, in (2233) the quantum behaviour maximally violating the CGLMP inequality \cite{cglmp} was given in \cite{acin,NPA}. However, for it we find that $||P||_1\simeq1.98$, while the maximal value $||P||_1=2$ is attained for the behaviour that yields the maximal CGLMP value when restricted to a maximally entangled two-qutrit state \cite{cglmp,acin}. This leads to consider whether for every bipartite maximally entangled state of dimension $d$,
\begin{equation}\label{maxent}
|\Phi^+_d\rangle=\frac{1}{\sqrt{d}}\sum_{j=0}^{d-1}|jj\rangle,
\end{equation}
there exists measurements such that the corresponding behaviour attains the bound of Theorem 1. In the following we show that this is indeed the case $\forall d$. We will use the construction of \cite{cglmp} that provides measurements on a $(22dd)$ scenario (i.\ e.\ $x,y=1,2$ and $a,b=1,\ldots,d$) such that acting on $|\Phi^+_d\rangle$ lead to the behaviour
\begin{equation}\label{maxentbeh}
P_d(ab|xy)=\left(2d^3\sin^2\left[\pi(a-b+\alpha(x)+\beta(y))/d\right]\right)^{-1},
\end{equation}
where $\alpha(1)=0$, $\alpha(2)=1/2$, $\beta(1)=1/4$ and $\beta(2)=-1/4$. This construction is enough to show that for a maximally entangled state of any dimension, there exist a behaviour for which the bound of Theorem 1 can be saturated. However, it should be stressed that other behaviours arising from maximally entangled states can have this property as well.
\begin{theorem}
For any maximally entangled state $|\Phi^+_d\rangle$, the corresponding $(22dd)$ behaviour given in Eq.\ (\ref{maxentbeh}) attains the bound of Theorem 1: $||P_d||_1=2$ $\forall d$.
\end{theorem}
\begin{proof}
In order to obtain $||P_d||_1$ we will compute its $2d$ singular values, which we will denote by $\{\sigma_j^+,\sigma_j^-\}_{j=0}^{d-1}$ for reasons that will become clear later. To that aim we will first show that $P_d$ is normal (i.\ e.\ $P_d^\dag P_d=P_dP_d^\dag$) implying that the singular values correspond to the absolute values of the eigenvalues. By fixing the inputs $x,y$, the matrix of the behaviour (\ref{maxentbeh}) can be partitioned as
\begin{equation}\label{pd}
P_d=\left(
      \begin{array}{cc}
        P_{11} & P_{12} \\
        P_{21} & P_{22} \\
      \end{array}
    \right),
\end{equation}
where
\begin{align}
P(ab|11)&=P(ab|22)=\left(2d^3\sin^2\left[\pi(a-b+1/4)/d\right]\right)^{-1},\nonumber\\
P(ab|12)&=\left(2d^3\sin^2\left[\pi(a-b-1/4)/d\right]\right)^{-1},\nonumber\\
P(ab|12)&=\left(2d^3\sin^2\left[\pi(a-b+3/4)/d\right]\right)^{-1}.
\end{align}
Thus, to ease the notation we will rewrite Eq.\ (\ref{pd}) as
\begin{equation}\label{pd2}
P_d=\left(
      \begin{array}{cc}
        A & B \\
        C & A \\
      \end{array}
    \right).
\end{equation}
Notice that the $d\times d$ matrices $A,B,C$ are circulant, as for every fixed $x,y$ it holds that
\begin{equation}
P_d(ab|xy)=P_d(a+1,b+1|xy)
\end{equation}
where it should be understood here that $d+1=1$. We will use the following properties of $d\times d$ circulant matrices \cite{circ}: they all have the same eigenvectors $(1,\omega_j,\omega_j^2,\ldots,\omega_j^{d-1})^T/\sqrt{d}$ corresponding to eigenvalues
\begin{equation}\label{eign}
\lambda_j(P_{xy})=\sum_{b=1}^d P(1b|xy)\omega_j^{b-1},\quad j=0,1,\ldots,d-1,
\end{equation}
where $\omega_j=\exp(2\pi ij/d)$ are the $d$-th roots of unity. This particularly implies that all circulant matrices are normal and commute with each other. Hence, it is easy to check that $P_d$ is normal if $B^\dag B=C^\dag C$. To see that this is indeed the case, consider that Eq.\ (\ref{eign}) tells us that the eigenvalues of our matrices are given by
\begin{align}
\lambda_j(A)&=\frac{1}{2d^3}\sum_{k=0}^{d-1}\frac{\omega_j^k}{\sin^2\left(\frac{\pi k}{d}-\frac{\pi}{4d}\right)},\nonumber\\
\lambda_j(B)&=\frac{1}{2d^3}\sum_{k=0}^{d-1}\frac{\omega_j^k}{\sin^2\left(\frac{\pi k}{d}+\frac{\pi}{4d}\right)},\nonumber\\
\lambda_j(C)&=\frac{1}{2d^3}\sum_{k=0}^{d-1}\frac{\omega_j^k}{\sin^2\left(\frac{\pi k}{d}-\frac{3\pi}{4d}\right)}.
\end{align}
These summation formulas are computed in the Appendix, obtaining
 \begin{align}
\lambda_j(A)&=-\frac{ie^{i\frac{\pi j}{2d}}}{d^2}(j+i(d-j)),\nonumber\\
\lambda_j(B)&=\frac{ie^{-i\frac{\pi j}{2d}}}{d^2}(j-i(d-j)),\nonumber\\
\lambda_j(C)&=\frac{ie^{i\frac{3\pi j}{2d}}}{d^2}(j-i(d-j)).\label{eigs}
\end{align}
Notice then that $|\lambda_j(B)|=|\lambda_j(C)|$ corresponding to the same eigenvector, which in addition to the fact that these matrices are normal, implies indeed that $B^\dag B=C^\dag C$. Hence, $P_d$ is normal and, therefore, its singular values are the absolute value of its eigenvalues. We compute now then the latter. Since all circulant matrices are diagonalized by the same unitary $U$ (i.\ e.\ $X=UD_XU^\dag$ with $D$ diagonal for every circulant matrix $X$), it holds then that the matrix
\begin{equation}
D=\left(
  \begin{array}{cc}
    D_A & D_B \\
    D_C & D_A \\
  \end{array}
\right)=\left(
  \begin{array}{cc}
    U^\dag & 0 \\
    0 & U^\dag \\
  \end{array}
\right)P_d\left(
  \begin{array}{cc}
    U & 0 \\
    0 & U \\
  \end{array}
\right)
\end{equation}
has the same eigenvalues as $P_d$. Using now the Schur complement condition that tells us that
\begin{equation}
\det(D)=\det(D_A)\det(D_A-D_BD_A^{-1}D_C),
\end{equation}
we have that
\begin{equation}
\lambda_j(A)-\lambda_j(P_d)-\frac{\lambda_j(B)\lambda_j(C)}{\lambda_j(A)-\lambda_j(P_d)}=0,
\end{equation}
and, hence,
\begin{equation}
\lambda_j^\pm(P_d)=\lambda_j(A)\pm\sqrt{\lambda_j(B)\lambda_j(C)}.
\end{equation}
Thus, using Eqs. (\ref{eigs}), we obtain that
\begin{equation}
\sigma_j^\pm(P_d)=|\lambda_j^\pm(P_d)|=\frac{|-(j+i(d-j))\pm(j-i(d-j))|}{d^2},
\end{equation}
which means that
\begin{equation}
\sigma_j^-(P_d)=\frac{2j}{d^2},\quad\sigma_j^+(P_d)=\frac{2(d-j)}{d^2}.
\end{equation}
Hence, we finally obtain that
\begin{equation}
||P_d||_1=\sum_{j=0}^{d-1}(\sigma_j^+(P_d)+\sigma_j^-(P_d))=2.
\end{equation}
\end{proof}

This result is interesting because it shows that the bound of Theorem 1 is attainable in $\q\backslash\l$ and that behaviours arising from maximally entangled states are extremal in this sense. Notwithstanding, it has some further application. If a real square matrix $P$ has singular value decomposition given by
\begin{equation}
P=\sum_i\sigma_iu_iv_i^T=U\Sigma V^T,
\end{equation}
then
\begin{equation}
||P||_1=\tr\Sigma=\tr(PO)
\end{equation}
with $O=VU^T$ orthogonal (as so are $U$ and $V$). Thus, by choosing $G=O^T$, we can always construct a Bell expression such that for any given $\P$ it holds that $\langle P,G\rangle=||P||_1$. Remarkably, if we happen to have a quantum behaviour such that $||P||_1=\sqrt{m_Am_B}$ (i.\ e.\ it saturates the bound of Theorem 1), then the aforementioned prescription immediately yields a Bell expression which is maximized in $\q$ by $P$. This is because $||G||_\infty=1$ and, hence, by inequality (\ref{ineq1}), there cannot exist any other quantum behaviour $\textbf{R}$ such that $\langle R,G\rangle>\sqrt{m_Am_B}$. Thus, this allows to construct Bell inequalities which are maximally violated in $\q$ by different behaviours of interest. Theorem 6 shows that this is possible for maximally entangled states of any dimension \cite{sos}. For example, following this prescription for the behaviour (\ref{maxentbeh}) with $d=3$ leads to a Bell expression $G(\Phi^+_3)$ which is then maximized in $\q$ by $|\Phi^+_3\rangle$ and is given by the following coefficient matrix \footnote{In order to have a symmetrical Bell expression we have used the reduced singular value decomposition of $P$. Hence, $U$ and $V$ are not orthogonal (they are not square) and nor is $G$. Still, it holds that $G(\Phi^+_3)_\q=2$ as $||G||_\infty=1$.}
\begin{equation}
G=\left(
                  \begin{array}{cccccc}
                    \frac{2+\sqrt{3}}{6} & \frac{2-\sqrt{3}}{6} & -\frac{1}{6} & \frac{2+\sqrt{3}}{6} & -\frac{1}{6} & \frac{2-\sqrt{3}}{6} \\
                    -\frac{1}{6} & \frac{2+\sqrt{3}}{6} & \frac{2-\sqrt{3}}{6} & \frac{2-\sqrt{3}}{6} & \frac{2+\sqrt{3}}{6} & -\frac{1}{6} \\
                    \frac{2-\sqrt{3}}{6} & -\frac{1}{6} & \frac{2+\sqrt{3}}{6} & -\frac{1}{6} & \frac{2-\sqrt{3}}{6} & \frac{2+\sqrt{3}}{6} \\
                    \frac{2+\sqrt{3}}{6} & -\frac{1}{6} & \frac{2-\sqrt{3}}{6} & -\frac{1}{6} & \frac{2+\sqrt{3}}{6} & \frac{2-\sqrt{3}}{6} \\
                    \frac{2-\sqrt{3}}{6} & \frac{2+\sqrt{3}}{6} & -\frac{1}{6} & \frac{2-\sqrt{3}}{6} & -\frac{1}{6} & \frac{2+\sqrt{3}}{6} \\
                    -\frac{1}{6} & \frac{2-\sqrt{3}}{6} & \frac{2+\sqrt{3}}{6} & \frac{2+\sqrt{3}}{6} & \frac{2-\sqrt{3}}{6} & -\frac{1}{6} \\
                  \end{array}
                \right)^T.
\end{equation}
Notice that this Bell inequality separates $\l$ from $\q$ as it is straightforward to find that the maximal value of $G(\Phi^+_3)$ under $\l$ is $(3\sqrt{3}+5)/6\simeq1.70$. A different example of a Bell inequality maximally violated by $|\Phi^+_3\rangle$ can be found in \cite{liang}.

Thus, Theorem 6 also shows that that for a maximally entangled state of any dimension there always exists a Bell inequality that is maximally violated in $\q$ by it and how to construct it \cite{son}. In this sense, one can then devise a DIQIP protocol for which maximally entangled states are optimal within $\q$. This might also be of relevance in the context of self-testing \cite{yao} if it turned out that the behaviour (\ref{maxentbeh}) is the only one maximizing these Bell expressions in $\q$. Self-testing arises when a certain behaviour is the unique to attain a particular Bell value. 
 This allows to check the performance of a quantum set-up without trusting any of the devices, particularly when it can be made robust \cite{rob,robm}
. Using the techniques of \cite{robm} with the Bell inequality $G(\Phi^+_3)$ and its generalizations for other dimensions, it could be possible to check whether robust self-testing of maximally entangled states is possible in this way. 

Notice, moreover, that this observation above that allows to construct a Bell expression such that $\langle P,G\rangle=||P||_1$ can be used in other contexts. For instance, the results of Theorems 4 and 5 imply that for every PR box considered there we can write down constructively a Bell expression, and hence a potential DIQIP protocol, whose performance has a quantitative gap with any quantum behaviour.

\section{Correlation scenarios}\label{corrsec}

If we restrict ourselves ourselves to two-outcome scenarios ($d_A=d_B=2$) and taking $a,b\in\{-1,1\}$, all behaviours in $\ns$ can be alternatively characterized by the correlators
\begin{equation}
\langle A_xB_y\rangle=\sum_{ab} abP(ab|xy)
\end{equation}
and the marginal expectations
\begin{equation}
\langle A_x\rangle=\sum_aaP(a|x), \quad\langle B_y\rangle=\sum_bbP(b|y).
\end{equation}
As mentioned before, it has been shown in \cite{epping} that for the particular class of correlator Bell inequalities,
\begin{equation}\label{eqepping}
\sum_{xy}G_{xy}\langle A_xB_y\rangle\leq||G||_\infty\sqrt{m_Am_B}
\end{equation}
must hold $\forall \P\in\q$ and every real $m_A\times m_B$ matrix $G$. We will denote by $C$ the $m_A\times m_B$ matrix with entries $C_{xy}=\langle A_xB_y\rangle$. It can be shown that the result above can also be proved using similar techniques as in Theorems 1 and 2 \cite{wehner}. For that, one just needs to consider that behaviours in $Q^1$ must fulfill \cite{NPA} that a $(m_A+m_B)\times(m_A+m_B)$ real positive semidefinite matrix $\left(\begin{array}{cc}
\hat{Q} & C \\
C^T & \hat{R} \\
\end{array}\right)$ exists with $\hat{Q}_{xx}=1$ and $\hat{R}_{yy}=1$ and proceed as in the proofs of Theorems 1 and 2 to upper bound $\tr(CG^T)$. This does not only provide an alternative proof of inequality (\ref{eqepping}) but it also shows that this bound cannot give stronger constraints than $Q^1$. Moreover, as in Theorems 1 and 2 this leads to the following condition
\begin{equation}\label{ineqc1}
||C||_1\leq\sqrt{m_Am_B}\quad\forall \P\in\q.
\end{equation}
However, it turns out that this condition is strictly weaker than Theorem 1 (i.\ e.\ every behaviour detected as non-quantum by the above inequality is also non-quantum by Theorem 1) as we show in the following.
\begin{proposition}
For every two-outcome behaviour $\P\in\ns$, if $||C||_1>\sqrt{m_Am_B}$, then $||P||_1>\sqrt{m_Am_B}$.
\end{proposition}
\begin{proof}
Using again the mapping from $P$ to $P'$ as in the proof of Theorem 4, we obtain that
\begin{align}
||P||_1&=\left|\left|\left(
            \begin{array}{cc}
              P'_{11} & P'_{12} \\
              P'_{21} & P'_{22} \\
            \end{array}
          \right)\right|\right|_1\nonumber\\
          &\geq\frac{1}{2}(||\sum_{ab}P'_{ab}||_1+||P'_{11}-P'_{12}-P'_{21}+P'_{22}||_1)\nonumber\\
          &=\frac{1}{2}\left[\left|\left|\left(
                                    \begin{array}{ccc}
                                      1 & \cdots & 1 \\
                                      \vdots & \ddots & \vdots \\
                                      1 & \cdots & 1 \\
                                    \end{array}
                                  \right)\right|\right|_1+||C||_1\right]\nonumber\\
                                  &=\frac{\sqrt{m_Am_B}+||C||_1}{2},
\end{align}
where the inequality comes from Corollary 3 in \cite{bani}.
\end{proof}

The above result agrees with intuition since with $C$ we are not fully characterizing $\P$. This suggests to consider on the analogy of Theorem 2 a condition including the marginals $\langle A_x\rangle$ and $\langle B_y\rangle$. Indeed, using similar arguments based on the first step of the NPA hierarchy one can show that, for every matrix $G\in\mathbb{R}^{m_B\times m_B}$, it must hold for every quantum behaviour that
\begin{align}
&\sum_{xy}G_{xy}(\langle A_xB_y\rangle-\langle A_x\rangle\langle B_y\rangle)\leq||G||_\infty\sqrt{m_Am_B}\nonumber\\ &-\frac{||G||_\infty}{2}\sqrt{\frac{m_B}{m_A}}\sum_{x}\langle A_x\rangle^2-\frac{||G||_\infty}{2}\sqrt{\frac{m_A}{m_B}}\sum_{y}\langle B_y\rangle^2.\label{ineqc2}
\end{align}
This particularly implies the following.
\begin{theorem}
For every two-outcome behaviour $\P\in\q$ it holds that
\begin{equation*}
||C'||_1\leq\sqrt{m_Am_B}\left(1-\sum_{x}
\frac{\langle A_x\rangle^2}{2m_A}-\sum_{y}\frac{\langle B_y\rangle^2}{2m_B}\right),
\end{equation*}
where $C'$ has entries $C'_{xy}=\langle A_xB_y\rangle-\langle A_x\rangle\langle B_y\rangle$.
\end{theorem}
\begin{proof}
We have to show that inequality (\ref{ineqc2}) is true for quantum behaviours. This follows from the fact that $Q^1$ is equivalent \cite{NPA} to the positive semi-definiteness of
$$\left(
    \begin{array}{ccc}
      1 & \langle A\rangle & \langle B\rangle \\
      \langle A\rangle^T & \hat{Q} & C \\
      \langle B\rangle^T & C^T & \hat{R} \\
    \end{array}
  \right),$$
where $\langle A\rangle$  ($\langle B\rangle$) is an $m_A$($m_B$)-dimensional vector with entries $\langle A_x\rangle$ ($\langle B_y\rangle$). By Schur's complement condition this leads to
$$\left(\begin{array}{cc}
\hat{Q} & C \\
C^T & \hat{R} \\
\end{array}\right)-\left(
                     \begin{array}{c}
                       \langle A\rangle \\
                       \langle B\rangle \\
                     \end{array}
                   \right)\left(
                            \begin{array}{cc}
                              \langle A\rangle^T & \langle B\rangle^T \\
                            \end{array}
                          \right)=\left(\begin{array}{cc}
\hat{Q'} & C' \\
C'^T & \hat{R'} \\
\end{array}\right)\geq0$$
with $\hat{Q'}_{xx}=1-\langle A_x\rangle^2$ and $\hat{R'}_{yy}=1-\langle B_y\rangle^2$. Proceeding as in the proof of Theorems 1 and 2 we can upper bound $\tr(C'G^T)$ to obtain the desired result.
\end{proof}

\section{Conclusions}

We have shown that the first step of the NPA hierarchy allows to obtain simple analytical conditions constraining the set of quantum behaviours in general bipartite Bell scenarios, whose strength and non-triviality have been illustrated. Since not all problems in quantum nonlocality and DIQIP can be addressed numerically, we expect these conditions to be of utility, filling the hitherto lack of such general tools. In fact, we have applied these conditions here to obtain a variety of results. In Sec.\ \ref{nontsec} we have shown that our bounds are tight and, in general, non-trivial, and we have used them to prove that there exists a finite gap (whose size we have estimated) between the quantum set and PR boxes in several general scenarios answering a question raised in \cite{ravi}. In Sec.\ \ref{tsisec} we have provided a systematic construction of quantum bounds for arbitrary Bell inequalities and we have shown that for a maximally entangled state of any dimension one can obtain a behaviour that attains the bound of Theorem 1. Interestingly, this can be translated into a Bell inequality whose Tsirelson bound is reached by such a state. This could be applied for robust self-testing of maximally entangled states using the techniques developed in \cite{robm}. Finally, in Sec.\ \ref{corrsec} we studied the particular case of correlation scenarios and established some links with the results of \cite{epping}. Several other ideas will be further investigated in the future. Given a Bell inequality, it would be interesting to find a procedure to find the best form of $G$ in (\ref{ineq1}) to obtain its quantum upper bound and when it can be optimal. It is also worth further research to characterize which behaviours in $\q\backslash\l$ attain the bound in Theorem 1: do they only arise from maximally entangled states? It would be also desirable to extend this approach to the multipartite setting. Last, it is worth studying whether stronger analytical conditions as those derived here can be obtained by considering further steps of the NPA hierarchy.

\begin{acknowledgments}
I thank A. Acin and his collaborators for sharing with me a draft of their work on Bell inequalities maximally violated by maximally entangled states \cite{sos}. This research was funded by the Spanish MINECO through
grants MTM 2010-21186-C02-02, MTM2011-26912 and MTM2014-54692 and the CAM regional research consortium QUITEMAD+CM S2013/ICE-2801.
\end{acknowledgments}

\begin{appendix}

\section{Summation formulas}

Here we prove the summation formulas used in the proof of Theorem 6:
\begin{align}
\frac{1}{2d^3}\sum_{k=0}^{d-1}\frac{\omega_j^k}{\sin^2\left(\frac{\pi k}{d}-\frac{\pi}{4d}\right)}&=-\frac{ie^{i\frac{\pi j}{2d}}}{d^2}(j+i(d-j)),\nonumber\\
\frac{1}{2d^3}\sum_{k=0}^{d-1}\frac{\omega_j^k}{\sin^2\left(\frac{\pi k}{d}+\frac{\pi}{4d}\right)}&=\frac{ie^{-i\frac{\pi j}{2d}}}{d^2}(j-i(d-j)),\nonumber\\
\frac{1}{2d^3}\sum_{k=0}^{d-1}\frac{\omega_j^k}{\sin^2\left(\frac{\pi k}{d}-\frac{3\pi}{4d}\right)}&=\frac{ie^{i\frac{3\pi j}{2d}}}{d^2}(j-i(d-j)),
\end{align}
which follow from
\begin{align}
S_\theta=&\sum_{k=0}^{d-1}\frac{\omega_j^k}{\sin^2\left(\pi(k+\theta)/d\right)}\nonumber\\
&=-\frac{4de^{-i2\pi j\theta/d}}{(1-e^{-i2\pi\theta})^2}\left(j+e^{-i2\pi\theta}(d-j)\right),\quad\theta\notin\mathbb{Z}.\label{eqapp}
\end{align}
To verify Eq.\ (\ref{eqapp}), first notice that the geometric sum yields
\begin{equation}
\sum_{n=0}^{d-1}e^{-i2\pi(k+\theta)n/d}=\frac{(1-e^{-i2\pi\theta})e^{i\pi(k+\theta)/d}}{2i\sin(\pi(k+\theta)/d)},
\end{equation}
and, therefore,
\begin{equation}
\sin(\pi(k+\theta)/d)=\frac{2i}{1-e^{-i2\pi\theta}}\sum_{n=0}^{d-1}e^{-i2\pi(k+\theta)(n+1/2)/d}.
\end{equation}
Thus, we find that
\begin{align}
S_\theta&=-\frac{4}{(1-e^{-i2\pi\theta})^2}\nonumber\\
&\times\sum_{m,n=0}^{d-1}e^{-i2\pi\theta(m+n+1)/d}\sum_{k=0}^{d-1}e^{-i2\pi(j-m-n-1)k/d}.
\end{align}
Notice that the inner sum is equal to zero, unless $j-m-n-1=zd$ for any integer $z$ for which the sum is equal to $d$. Given the values the indices take, this can only happen for $z=0$ or $z=-1$, hence obtaining
\begin{align}
S_\theta&=-\frac{4d}{(1-e^{-i2\pi\theta})^2}\nonumber\\
&\times\sum_{m,n=0}^{d-1}e^{-i2\pi\theta(m+n+1)/d}(\delta_{m+n,j-1}+\delta_{m+n,j-1+d})\nonumber\\
&=-\frac{4d}{(1-e^{-i2\pi\theta})^2}\left(je^{-i2\pi j\theta/d}+(d-j)e^{-i2\pi (j+d)\theta/d}\right)
\end{align}
that leads to the desired result.
\end{appendix}

\end{document}